\documentclass[runningheads]{llncs}
\usepackage{amsmath,amssymb}
\usepackage[ruled,vlined]{algorithm2e}
\usepackage{hyperref}
\usepackage{graphicx}
\usepackage{array}
\usepackage{caption}
\begin{document}
\title{Time Optimal Gathering of Myopic Robots on an Infinite Triangular Grid}

\titlerunning{Time Optimal Gathering of Myopic Robots on Infinite Triangular Grid }

\author{Pritam Goswami\inst{1}\orcidID{0000-0002-0546-3894\thanks{The first three authors are full time research scholars in Jadavpur University.}} \and
Avisek Sharma\inst{1}\orcidID{0000-0001-8940-392X}\and
Satakshi Ghosh\inst{1}\orcidID{0000-0003-1747-4037}\and
Buddhadeb Sau\inst{1}\orcidID{0000-0001-7008-6135}}
\authorrunning{P.Goswami et al.}
%
\institute{Jadavpur University, 188, Raja S.C. Mallick Rd,
Kolkata 700032, India\\
\email{\{pritamgoswami.math.rs, aviseks.math.rs, satakshighosh.math.rs,  buddhadeb.sau\}@jadavpuruniversity.in}}

\maketitle              
\begin{abstract}
This work deals with the problem of gathering $n$ oblivious mobile entities, called robots, at a point (not known beforehand) placed on an infinite triangular grid. The robots are considered to be myopic, i.e., robots have limited visibility. Earlier works of gathering mostly considered the robots either on a plane or on a circle or on a rectangular grid under both full and limited visibility. In the triangular grid, there are two works to the best of our knowledge. The first one is by Cicerone et al. on arbitrary pattern formation where full visibility is considered. The other one by Shibata et al. which considers seven robots with 2- hop visibility that form a hexagon with one robot in the center of the hexagon in a collision-less environment under a fully synchronous scheduler .

In this work, we first show that gathering on a triangular grid with 1-hop vision of robots is not possible even under a fully synchronous scheduler if the robots do not agree on any axis. So one axis agreement has been considered in this work (i.e., the robots agree on a direction and its orientation). We have also shown that the lower bound for time is $\Omega(n)$ epochs when $n$ number of robots are gathering on an infinite triangular grid. An algorithm is then presented where a swarm of $n$ number of robots with 1-hop visibility can gather within $O(n)$ epochs under a semi-synchronous scheduler. So the algorithm presented here is time optimal.  
\keywords{Gathering  \and Triangular Grid \and Swarm robot \and Limited Visibility.}
\end{abstract}
\section{Introduction}
\subsection{Background  and motivation}
A swarm of robots is a collection of a large number of robots with minimal capabilities. In the present research scenario on robotics, researchers are interested in these swarms of robots as these inexpensive robots can collectively do many tasks which earlier were done by single highly expensive robots with many capabilities. The wide application of these swarm of robots in different fields (e.g., search and rescue operations, military operations, cleaning of large surfaces, disaster management, etc.) grabbed the interest of researchers in the field of swarm robotics.

The goal of the researches in this field is to find out the minimum capabilities the robots need to have to do some specific tasks like \textit{gathering} (\cite{DSKN16,F2005}) ,\textit{dispersion} (\cite{AugustineM18}), \textit{arbitrary pattern formation} (\cite{BoseAKS20}) etc. These capabilities are considered when modeling a robot for some specific task. Some of the well known robot models are $\mathcal{OBLOT}$, $\mathcal{FSTA}$, $\mathcal{FCOM}$ and $\mathcal{LUMI}$. In each of these models, the robots are \textit{autonomous} (i.e. there is no central control for the robots), \textit{anonymous} (i.e. the robots do not have any unique identifier), \textit{homogeneous} (i.e all the robots upon activation execute the same deterministic algorithm), \textit{identical} (i.e robots are physically identical). In the $\mathcal{OBLOT}$ model the robots are considered to be  \textit{silent}(i.e robots do not have any direct means of communication) and oblivious (i.e. the robots do not have any persistent memory so that they can remember their earlier state). In $\mathcal{FSTA}$ model, the robots are \textit{silent} but not\textit{ oblivious}. In $\mathcal{FCOM}$, the robots are not \textit{silent} but are \textit{oblivious}. And in $\mathcal{LUMI}$ model, the robots are neither \textit{silent} nor\textit{ oblivious}. Their are many works that have been done considering these four robot models (\cite{SB15,BoseAKS20,BoseKAS21,DSN14,F2005,LunaFPS21}). In this paper, we have considered the weakest $\mathcal{OBLOT}$ model, among these four models.

The activation time of the robots is a huge factor when it comes to the robots doing some tasks. A scheduler is said to be controlling the activation of robots. Mainly there are three types of schedulers that have been used vastly in literature. \textit{Fully synchronous} (FSYNC) scheduler where the time is divided into global rounds of the same length and each robot is activated at the beginning of each round, \textit{semi-synchronous} (SSYNC) scheduler where time is divided into equal-length rounds but all robots may not be activated at the beginning of each round and \textit{asynchronous} (ASYNC) scheduler where any robot can get activated any time as there is no sense of global rounds. Among these, FSYNC and  SSYNC  schedulers are considered to be less practical than ASYNC scheduler. Still, it has been used in many works (\cite{LunaFPS21}) as providing algorithms for a more general and more realistic ASYNC scheduler is not always easy. In this paper, we have considered the SSYNC scheduler.

Vision is another important capability that robots have. The vision of a robot acquires information about the positions of other robots in the environment. A robot can have either full or restricted visibility. In \cite{SB15,BKAS18,CSN18,DSKN16,DSN14,KMP08} authors have modeled the robots to have infinite or full vision. The biggest drawback of full vision is that it is not possible in practical applications due to hardware limitations. So in \cite{AOSY99,F2005,PoudelS21}, authors  considered \textit{limited visibility}. A robot with limited visibility is called a myopic robot. A myopic robot on the plane is assumed to see only up to a certain distance called \textit{visibility range}. In graphs though, the vision of a robot is assumed to be all the vertices within a certain hop from the vertex on which the robot is located. Other than \textit{limited vision}, robots can have\textit{ obstructed vision} where even if the vision of a robot is infinite it might get obstructed by other robots in front of it. This model is also more practical than using point robots that can see through other robots. So in \cite{BoseKAS21}, obstructed vision model has been considered. 

In this paper, we are interested in the problem of \textit{gathering}. The \textit{gathering problem} requires a swarm of robots that are placed either on a plane or on a graph,  to move to a single point that is not known to the robots a priori (Ideal Gathering Configuration). In this work, we have considered the robots on an infinite triangular grid having the least possible vision of 1-hop  under the SSYNC scheduler. Our solution also works under obstructed vision model as a robot needs no information about other robots who are not directly adjacent to it.

Earlier Gathering problem has been considered under limited vision on the plane (\cite{F2005}), but movements of robots are not restricted in the plane as there are infinitely many paths between any two points on a plane. So it would have been interesting to consider this problem on discrete terrain where the movement of a robot is restricted. 
And since grid network has wide application in various fields it was natural to study this problem under different kinds of grids. Now an infinite regular tessellation grid is one of the 3 types of infinite regular grids namely, infinite square grid, infinite triangular grid, and infinite hexagonal grid (\cite{vn2577477}). Our goal is to solve this problem for any infinite regular tessellation grid with the least possible vision for a robot. In \cite{PoudelS21} a solution has already been provided by the authors where the terrain is an infinite square grid embedded on a plane and the robots can either move diagonally from one grid point to another or moves along the edges of the grid. But in their work, the robots can see up to a distance of 2 units (each edge length of the grid is considered to be one unit). So in this paper, by providing a solution for the infinite triangular grid, where a robot can see only up to a unit of distance,  we reached a little closer to our goal of providing a solution for this problem for any infinite regular tessellation grid. Furthermore, we also drew motivation for framing this problem for an infinite triangular grid from the application perspective of it. In \cite{ZhangH05}, authors have shown that for some robots with sensors the coverage will be maximum if the robots are forming a triangular grid and the length of each edge is $\sqrt{3}s$ where $s$ is the sensing radius for the sensors on the robots. So coverage wise triangular grid is better than any other regular tessellation grid. For these specific reasons, we have considered this problem on this specific terrain.

The literature on this problem is very rich. In the following subsection, we have provided a glimpse of the rich literature that lead us to write this paper.
\subsection{Earlier works}
In this paper, we are specifically focused on the problem of \textit{gathering}. Earlier the problem was mainly studied considering the robots on a plane (\cite{AP06,SB15}). But currently, many researchers have been interested in gathering on the discrete environment as well, (\cite{DSKN16,DSN14,KKN10,KMP08}) as movements in graphs become more restricted which is practical in real-life scenarios. In \cite{KMP08}, Klasing et al. studied the gathering problem on a ring and proved that it is impossible to gather on a ring without multiplicity detection. In \cite{DSKN16}, D’Angelo et al. first characterized the problem of gathering on a tree and finite grid. He  has proved that gathering even with global-strong multiplicity detection is impossible if the configuration is periodic or, symmetric
with the line of symmetry passing through the edges of the grid.

Another capability of these robots is their vision. After activation, a robot takes a snapshot of its surroundings to collect information about the positions of the other robots. Gathering has been studied extensively where robots are assumed to have full or infinite visibility (\cite{AP06,SB15,BKAS18,CSN18,CSN19,DSKN16,DSN14,KKN10,KMP08}). But in the application, it is impossible due to hardware limitations. So, in \cite{AOSY99} Ando et al. provided an algorithm where indistinguishable robots with a limited vision on a plane without any common coordinate system converge to a point under a semi-synchronous scheduler. In \cite{F2005} Flocchini et al. have produced a procedure that guides robots with a limited vision on a plane to gather at a single point in finite time. In their work, they have assumed the robots have agreement on the direction and orientation of the axes under an asynchronous scheduler. In \cite{L20}, the authors have shown that gathering is possible by robots on a circle with agreement on the clockwise direction even if a robot can not see the location at an angle $\pi$ from it, under a semi-synchronous scheduler. Gathering under limited visibility where the robots are placed in a discrete environment
has been recently studied by the authors in \cite{PoudelS21} where algorithms have been provided with both one and two-axis agreement under viewing range 2 and 3 simultaneously and square connectivity range $\sqrt{2}$  under asynchronous scheduler.
\subsection{Our contribution}
Recently, in \cite{C21}, the authors have provided an algorithm for robots on a triangular grid to form any arbitrary pattern from any asymmetric initial configuration. In their work, they have assumed that the target configuration can have multiplicities also. So the algorithm provided in \cite{C21} can be used for gathering where the target configuration contains only one location for each robot. But in their work, they have assumed the robots have full visibility which is impractical as in application robots can't have an infinite vision. Also, their algorithm works only when the initial configuration is asymmetric. 

Considering limited vision this problem has earlier been done in the euclidean plane in \cite{F2005}. But in the plane, the movement of a robot is not at all restricted as there are infinitely many paths between any two points on the plane. Also, the authors have considered two axis agreement which makes the robot more powerful which is against the motivation of research on swarm robot algorithms where we need to find the minimum capabilities for the robots to do a specific task. 

In \cite{PoudelS21}, the authors have presented a technique for gathering under limited visibility under an infinite rectangular grid. In their work, they have presented two algorithms. In the first algorithm, they have considered two axis agreement and a vision of $2 \times$ edge length of the grid. And in the second algorithm considering one axis agreement and vision of $3 \times$ edge length of the grid for any robot they have provided an algorithm where the robots may not gather but will surely be on a horizontal segment of unit length (Relaxed Gathering Configuration). Both of these algorithms are not collision-free. Also, observe that none of their algorithms are able to gather if the visibility for each robot is $1\times$ edge length of the grid.

In our work, we have given a  characterization of the gathering problem of myopic robots with 1-hop vision on a triangular grid with any connected initial configuration. We have shown that myopic robots with 1-hop visibility on an infinite triangular grid which agree on the direction of both axis can not gather even under a fully synchronous scheduler if they do not agree on the orientation of any axis. So assuming that myopic robots on an infinite triangular grid have 1-hop (i.e. $1 \times$ length of an edge of the triangular grid) visibility and they agree on the direction and orientation of any one of the three lines that generate the infinite triangular grid, we have provided an algorithm \textsc{1-hop 1-axis gather} (Algorithm \ref{algo1}) which gathers these robots on a single grid point within $O(n)$ epochs under semi-synchronous scheduler where $n$ is the number of myopic robots on the grid. Where one epoch is a time interval such that within which each robot has been activated at least once. We have also shown that any gathering algorithm on a triangular grid must take $\Omega(n)$ epochs where $n$ is the number of robots placed on the infinite triangular grid. Therefore the algorithm we presented in this paper is asymptotically time optimal.

In the following Table \ref{tab:Table} we have compared our work with the works in \cite{C21}, \cite{PoudelS21} and \cite{F2005}.

\begin{table}[ht]
    \centering
    \begin{tabular}{ | m{2.2em} | m{2.5cm} | m{3cm}| m{3cm} | m{2.5cm} | } 
    \hline
     \textbf{SL. No.} & \textbf{Algorithm} & \textbf{Axis Agreement} & \textbf{Visibility} & \textbf{Ideal/Relaxed Gathering }\\
    \hline
    1& Algorithm in \cite{C21}& No axis agreement & Full visibility & Ideal\\
    \hline
    2& Algorithm in \cite{F2005} & Two axis & $V \in \mathbb{R}(> 0)$ & Ideal \\
    \hline
    3& $1^{st}$ Algorithm in \cite{PoudelS21} & Two axis & 2$\times$ edge length & Ideal\\
    \hline
    4 & $2^{nd}$ Algorithm in \cite{PoudelS21} & One axis & 3$\times$ edge length & Relaxed\\
    \hline
    5 & \textsc{1-hop 1-axis gather}\textbf{(This paper)}& One axis & 1$\times$ edge length & Ideal\\
    \hline
\end{tabular}
    \caption{Comparison table}
    \label{tab:Table}
\end{table}
\vspace{0.01\linewidth}

\section{Models and Definitions}
\subsection{Model}
An infinite triangular grid  $\mathcal{G}$ is a geometric graph where each vertex $v$ is placed on a plane and has exactly six adjacent vertices and any induced sub-graph $K_3$ forms an equilateral triangle. Let $R = \{r_1,r_2,r_3, \dots r_n\}$ be $n$ robots placed on the vertices of an infinite triangular grid $\mathcal{G}$.
\subsubsection{Robot Model:} The robots are considered to be-

\textbf{\textit{ autonomous}}: there is no centralized control.

\textbf{\textit{ anonymous}}: robots do not have any unique identifier (ID).

\textbf{\textit{ homogeneous}}: robots execute same deterministic algorithm.

\textbf{\textit{ identical}}: robots are identical by their physical appearance.

The robots are placed on the vertices of an infinite triangular grid $\mathcal{G}$ as a point. The robots do not have any multiplicity detection ability i.e., a robot can not decide if a vertex contains more than one robot or not. The robots do not agree on some global coordinate system, but each robot has its own local coordinate system with itself at the origin and handedness. However, the robots may agree on the direction and orientation of the axes. Based on that we consider the following model.

\textbf{\textit{One axis agreement model:}} In the one axis agreement model all robots agree on the direction and orientation of any specific axis.  Note that any vertex $v$ of the infinite triangular grid $\mathcal{G}$ is at the intersection of three types of lines. In this work, the robots will agree on the orientation and direction of any one of these three types of lines and consider it as its $y$-axis. Note that in this model the robots have a common notion of up and down but not about left or right.

As an input, a robot takes a snapshot after waking. This snapshot contains the position of other robots on $\mathcal{G}$ according to the local coordinates of the robot. In a realistic setting due to limitations of hardware, a robot might not see all of the grid points in a snapshot. So to limit the visibility of the robots we have considered the following visibility model.

\textbf{\textit{Visibility:}}  In $k$-hop visibility model, each robot $r$ can see all the grid points which are at most at a $k$-hop distance from $r$. In this paper, the robots are considered to have 1-hop visibility (i.e. $k=1$). Note that when $k=1$, a robot placed on a vertex $v$  of the infinite triangular grid $\mathcal{G} $ can only see the adjacent six vertices of $v$.

The robots operate in \textit{LOOK-COMPUTE-MOVE (LCM)} cycle. In each cycle a previously inactive or idle robot wakes up and does the following steps:

\textbf{\textit{LOOK:}} In this step after waking a robot placed on $u \in V$ takes a snapshot of the current configuration visible to it as an input. In this step, a robot gets the positions of other robots expressed under its local coordinate system.

\textbf{\textit{COMPUTE:}} In this step a robot computes a destination point $x$ adjacent to its current position, where $x \in V $ according to some deterministic algorithm with the previously obtained snapshot as input.

\textbf{\textit{MOVE:}} After determining a destination point $x \in V$ in the previous step the robot now moves to $x$ through the edge $ux \in E$. Note that if $x =u$ then the robot does not move. 

After completing one \textit{LCM} cycle a robot becomes inactive and again wakes up after a finite but unpredictable number of rounds and executes the \textit{LCM} cycle again.
\subsubsection{Scheduler Model:} Based on the activation and timing of the robots there are mainly three types of schedulers in the literature,

\textbf{\textit{Fully synchronous:}} In the case of a fully synchronous (FSYNC) scheduler time can be divided logically into global rounds where the duration of each round and each step of each round is the same. Also, each robot becomes active at the start of each round (i.e. the set of the active robot at the beginning of each round is the whole of $R$).

\textbf{\textit{Semi-synchronous:}} A semi-synchronous (SSYNC) scheduler is a more general version of a fully synchronous scheduler. In the case of a semi-synchronous scheduler, the set of active robots at the beginning of a round can be a proper subset of $R$. i.e. all the robots might not get activated at the beginning of a round. However, every robot is activated infinitely often.

\textbf{\textit{Asynchronous:}} An Asynchronous (ASYNC) scheduler is the most general model. Here a robot gets activated independently and also executes the \textit{LCM} cycles independently. The amount of time spent in each cycle and the inactive phase may be different for each robot and also for the same robot in two different cycles. This amount of time is finite but unbounded and also unpredictable. Hence there is no common notion of time. Moreover, a robot with delayed computation may compute at a time when other robots have already moved and changed the configuration. Thus the robot with delayed computation now computes with an obsolete configuration as input.

In this paper, we have considered the scheduler to be semi-synchronous. The scheduler that controls the time and activation of the robots can be thought of as an adversary. Observe that the semi-synchronous scheduler can be controlled as a fully synchronous scheduler as SSYNC is more general than FSYNC but not vice-versa. Also, an adversary can decide the local coordinate system of a robot (obeying the agreement rules of axes and orientation). However, after deciding on the coordinate system of a robot it can not be changed further. 
\subsection{Notations and definitions}
\begin{definition}
[Infinite Triangular Grid] An infinite triangular grid $\mathcal{G}$ is an infinite geometric graph $G = (V,E)$, where the vertices are placed on $\mathbb{R}^2$ having coordinates $\{(k,\frac{\sqrt3}{2}i):k\in\mathbb{Z}, i\in2\mathbb{Z}\}\cup\{(k+\frac12,\frac{\sqrt3}{2}i):k\in\mathbb{Z}, i\in2\mathbb{Z}+1\}$ and two vertices are adjacent if the euclidean distance between them is 1 unit. 
\end{definition}
It is to be noted that robots do not have access to this coordinates. This coordinates are used simply for describing the infinite triangular grid $\mathcal{G}$. 
\begin{definition}
[Configuration] A configuration formed by a set of robots $R$, denoted as $\mathcal{C}_R$ (or, simply $\mathcal{C}$) is the pair $(\mathcal{G},f)$ where, $f$ is a map from $V$ to $\{0,1\}$. For $v\in V$, $f(v)=1$ if and only if there is at least one robot on the vertex $v$.
\end{definition}
\begin{definition}[Visibility Graph]
 A visibility graph $G_{\mathcal{C}}$ for a configuration $\mathcal{C}=(\mathcal{G},f)$ is the sub graph of $\mathcal{G}$ induced by set of vertices $\{v\in V: f(v)=1\}$.
\end{definition}
It is not hard to produce a configuration $\mathcal{C}$ with disconnected $G_{\mathcal{C}}$, such that there exists no deterministic algorithm which can gather a set of robots starting from $\mathcal{C}$. So in this work, it is assumed that initially the visibility graph is connected and any algorithm that solves the gathering problem should maintain this connectivity during its complete execution.
\begin{definition}[\texttt{Extreme}]
A robot $r$ is said to be an \texttt{extreme} robot if the following conditions hold in its visibility:
\begin{enumerate}
    \item There is no other robot on the positive $y$-axis of $r$.
    \item Either left or right open half of $r$ is empty.
\end{enumerate}
\end{definition}

\subsection{Problem definition}
Suppose, a swarm of $n$ robots is placed on the grid points of an infinite triangular grid $\mathcal{G}$. The gathering problem requires devising an algorithm such that after some finite time all robots assemble at exactly one grid point and stay forever gathered at that grid point. 
\section{Impossibility Result}

\begin{figure}[ht]
    \centering
     \includegraphics[width=0.6\linewidth]{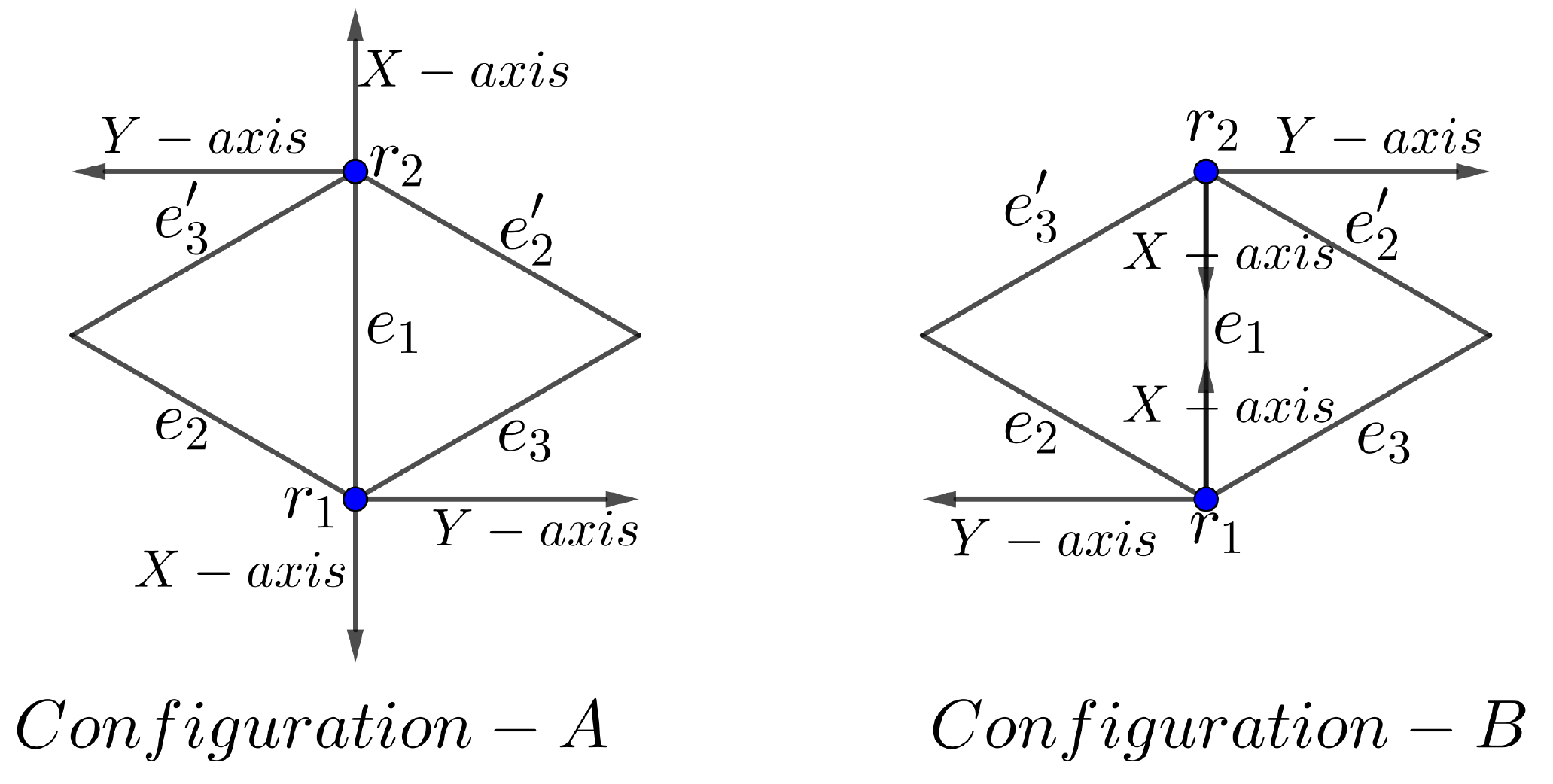}
     \caption{In the diagram the robots agree on direction of both the axes but do not agree on the orientation of the axes. }\label{Fig:impossible}
    \end{figure}

\begin{theorem}\label{thm:impo}
Gathering in a triangular grid is impossible without agreement on the orientation of any axis even when agreement on direction is present and under a fully synchronous scheduler and 1-hop visibility.
\end{theorem}
\begin{proof}
    
Let Configuration $A$ in Fig \ref{Fig:impossible} is the initial configuration. Let there be an algorithm $\mathcal{AL}$ on finite execution of which two robots  $r_1$ and $r_2$ from initial configuration $A$ gathers on a single vertex of the triangular grid $\mathcal{G}$. Let there is an adversary who has decided the direction and orientation of both $y$ and $x$ axes for both the robots as it is in Fig \ref{Fig:impossible}. Note that view of $r_1$ and $r_2$ is same and they agree on the direction of both axes in the diagram. So if $r_1$ moves through $e_1$ according to $\mathcal{AL}$ on activation, $r_2$ also moves through $e_1$. Observe that, if $r_1$ moves through any other edge other than $e_1$ then $r_2$ moves in such a way that the visibility graph becomes disconnected. So, both $r_1$ and $r_2$ moves through edge $e_1$ and the configuration transforms into configuration $B$. Using a similar argument it can be shown that Configuration $B$ can only transform into Configuration $A$ as the vision of the robots are 1-hop. So, a deadlock situation occurs. So, our assumption must be wrong. Thus we can conclude that there does not exist any algorithm $\mathcal{AL}$  on the execution of which robots on a triangular grid with a vision of 1-hop gather even under a fully synchronous scheduler without any axis agreement. To be more precise, even if the robots agree on the direction of the axes they will  still not gather when they do not have any agreement on the orientation of axes.\qed
\end{proof}

 Due to Theorem~\ref{thm:impo} we have considered one axis agreement model and devised an algorithm considering 1-hop vision under semi-synchronous scheduler.
\section{Gathering Algorithm}
\label{sec:algorithm}
In this section, an algorithm \textsc{1-hop 1-axis gather} (Algorithm   \ref{algo1}) is provided that will work for a swarm of $n$ myopic robots with one axis agreement and 1-hop visibility under a semi-synchronous scheduler. 
Note that under one axis agreement a robot can divide the grids into two halves based on the agreed line as the $y$-axis. An \texttt{extreme} robot $r$ will always have either left or right open half empty. Thus it is easy to see that when any one of the open halves is non-empty and $r$ is on a grid point $v$, two adjacent grid points of $v$ on the empty open half and another adjacent grid point of $v$ on $y$-axis and above $r$ will always be empty. In this situation, $r$ can uniquely identify the remaining three adjacent grid points of $v$ (one on the $y$-axis and below $r$ and the remaining two are on the non-empty half)  based on the different values of their $y-coordinates$. So an \texttt{extreme} robot can uniquely name them as $v_1, v_2$ and $v_3$ such that $y-coordinate$ of $v_i$ is less than $y-coordinate$ of $v_{i+1}$ and $i \in \{1, 2\}$ (Fig.\ref{Fig.uniname}). We denote position $v_j$ of an \texttt{extreme} robot $r$ as $v_j(r)$ where $j \in \{1, 2, 3\}$. Note that for a non-\texttt{extreme} robot $r$, there are two $v_2(r)$ and two $v_3(r)$ positions as $r$ have either both open halves empty or both open halves non empty.
\begin{figure}[b]
     \centering
     \includegraphics[width=0.25\linewidth]{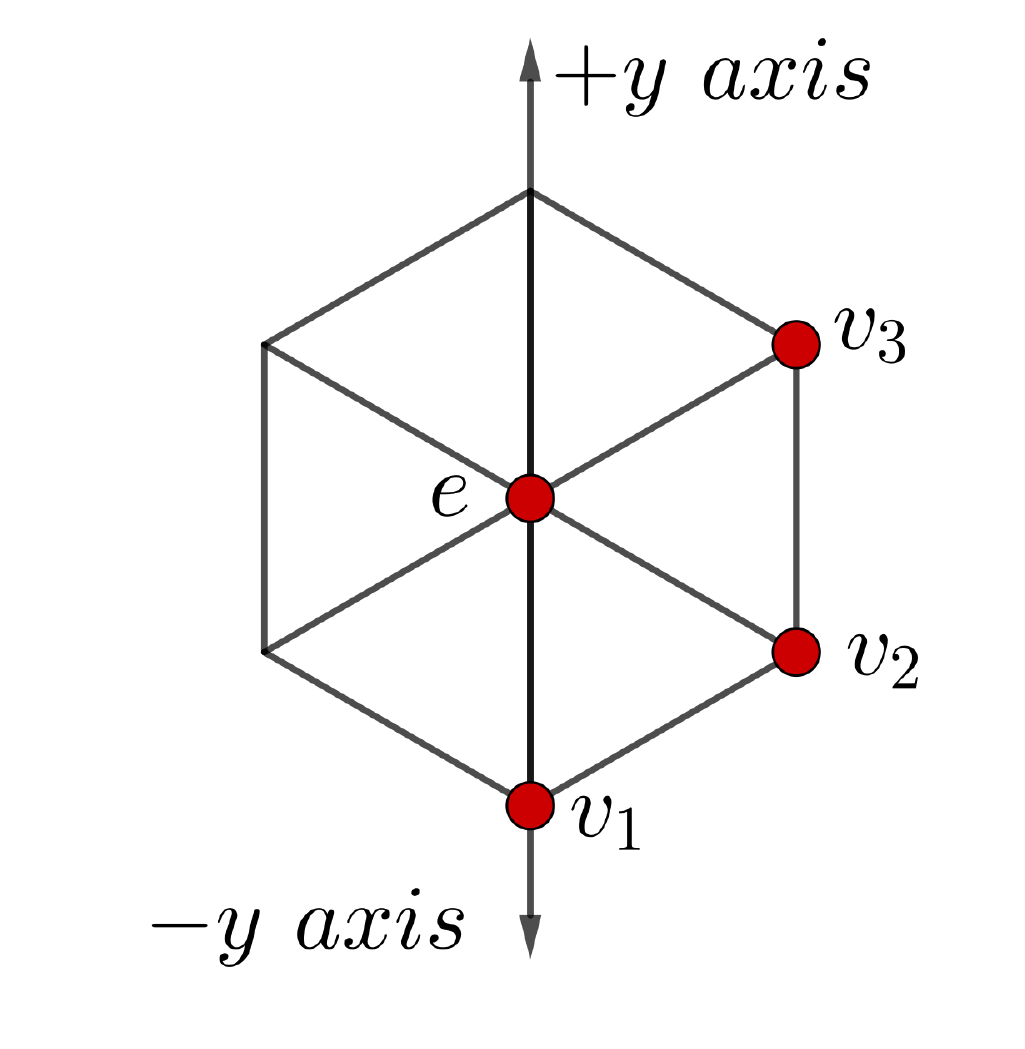}
     \caption{$e$ is an \texttt{extreme} robot it can uniquely identify the positions of $v_1$, $v_2$ and $v_3$ if it sees right or left open half non empty.}\label{Fig.uniname}
    \end{figure}
In the algorithm \textsc{1-hop 1-axis Gather} (\ref{algo1}), an \texttt{extreme} robot $r$ moves to $v_1(r)$ if there is a robot on $v_1(r)$ and there is no robot on $v_3(r)$. $r$ does not move when there is only a robot on $v_3(r)$ or there are robots only  on $v_3(r)$ and $v_1(r)$. In the other remaining cases, if $r$ sees at least one robot on the adjacent vertices it moves to $v_2(r)$. An \texttt{extreme} robot terminates when it does not see any other robot on the adjacent vertices. 

If $r$ is not an \texttt{extreme} robot, then it only moves if there is no robot with $y-coordinate$ greater than zero within its vision and there are two robots on both of its $v_2(r)$ positions. In this scenario the robot $r$ moves to $v_1(r)$.

In Fig. \ref{FigMoveExtreme} we have shown all possible views when a robot $r$ moves and in which direction it moves. In Fig. \ref{FigMoveExtreme} suppose a robot $r$ is placed on the node denoted by a black solid circle. The grid points that are encircled are occupied by other robots. For all the views of $r$ in $View-I$, $r$ moves to $v_1(r)$ and for all the views of $r$ in $View-II$, $r$ moves to $v_2(r)$.

\begin{algorithm}[H]
\small
\caption{\textsc{1-hop 1-axis gather} (for a robot $r$)}\label{algo1}
\KwData{Position of the robots on the adjacent grid points of $r$ on triangular grid $\mathcal{G}$.}
 \KwResult{A vertex on $\mathcal{G}$ adjacent to $r$, as destination point of $r$.}
 \If{$r$ is \texttt{extreme}}
    {
        \If{There is no robot on the adjacent grid points}
        {
            terminate\;   
        }
        \ElseIf{There is a robot only on $v_3(r)$ or there are robots only on both $v_1(r)$ and $v_3(r)$}
        {
            do not move\;
        }
        \ElseIf{There is a robot on $v_1(r)$ and no robot on $v_3(r)$}
        {
            move to $v_1(r)$\;
        }
        \Else
        {
            move to $v_2(r)$\;
        }
    }
    \Else{
        \If{There is a robot on both $v_2(r)$ and no robot on the vertices with $y-coordinate >0$ }
        {
            move to $v_1(r)$\;
        }
        \Else
        {
            do not move\;
        }
    }
    
\end{algorithm}

\begin{figure}[ht]
    \centering
     \includegraphics[width=0.7\linewidth]{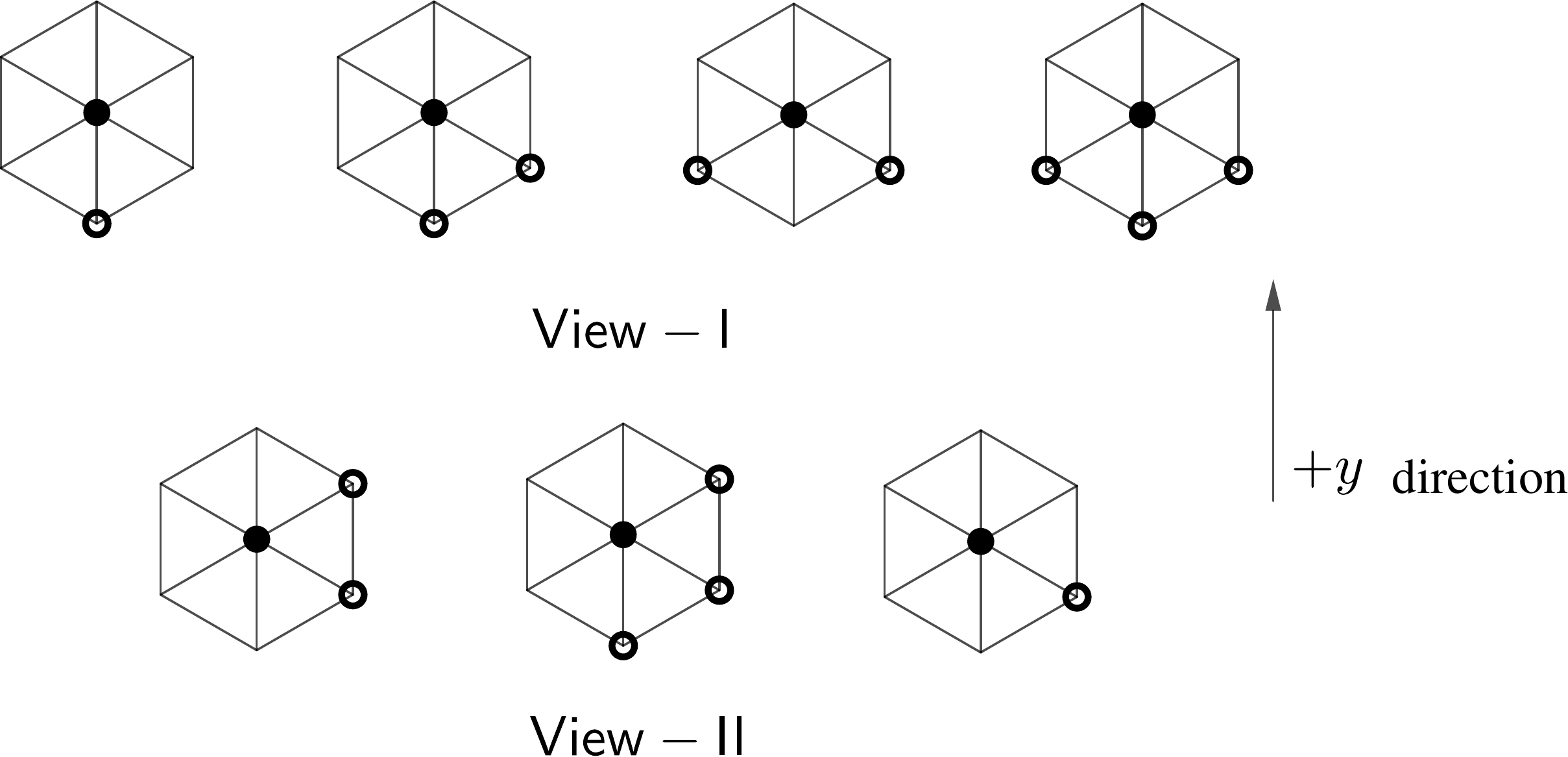}
     \caption{All possible views of a robot $r$ placed on a node indicated by a black solid circle when $r$ decides to move. Encircled point represents a robot occupied node For all views in $View-I$, $r$ moves to $v_1(r)$ position and for all views in $View-II$, $r$ moves to $v_2(r)$ position.  }
     \label{FigMoveExtreme}
    \end{figure}

\subsection{Correctness results:} The intuition of the algorithm \ref{algo1} is that the width of the configuration decreases while the visibility graph stays connected by the movement of the robots. The following results will make this intuition more concrete. Before that let us have some definitions which will be needed in the proof of the results.

\begin{definition}
[Layer] Let $H$ be a straight line perpendicular to the agreed direction of $y-$axis such that there is at least one robot on some grid points on $H$, then  $H$ is called a layer.
\end{definition}
\begin{definition}
[Top most layer, $H_t$] $H_t$ or top most layer of a configuration $\mathcal{C}$ is a layer such that there is no layer above it.
\end{definition}
\begin{definition}
[Vertical line, $L_v$] Let $L_v$ be a line that is parallel to the agreed direction of the $y$-axis such that there is at least one robot on some grid point on $L_v$, then $L_v$ is called a vertical line.
\end{definition}
\begin{definition}
[Left edge, $e_l$] Left edge of a configuration $\mathcal{C}$ or, $e_l$ is the vertical line such that there is no other vertical line on the left of $e_l$.
\end{definition}
\begin{definition}
[Right edge, $e_r$] Right edge of a configuration $\mathcal{C}$ or, $e_r$ is the vertical line such that there is no other vertical line on the right of $e_r$.
\end{definition}
\begin{definition}
[Width of a configuration $\mathcal{C}$] Width of a configuration $w(\mathcal{C})$ is defined as the distance between $e_l$ and $e_r$.
\end{definition}
\begin{definition}
[Depth of a vertical line $L_v$] Depth of a vertical line $L_v$ is defined as the distance between the layers $H_t$ and the layer on which the lowest robot on $L_v$ is located. We denote the depth of line $L_v$ as $d(L_v)$.
\end{definition}
Fig \ref{Figdef} shows all the entities of the above definitions. A brief overview of the correctness proof is given below along with the statements of the results.
\vspace{0.3cm}
\\
\textbf{Overview of the correctness proof:} In Lemma \ref{flemma1}, we have proved that the visibility graph will remain connected throughout the execution of the algorithm. It is necessary to prove this as otherwise, the robots may gather in several clusters on the infinite triangular grid.
Then we have shown that in Lemma \ref{flemma5} the width of the configuration will decrease in finite time. Now when the width of the configuration becomes one then there are only two vertical lines that contain robots. These lines are left edge $e_l$ and right edge $e_r$. Now in this scenario from Lemma \ref{flemma2} the robots on the topmost layer will always move below and the depth of both $e_l$ and $e_r$ never increases (by Lemma \ref{flemma4}). So the depth of both the right and left edge now decreases in each epoch. Hence within finite time, the depth will also become one for either $e_l$ or $e_r$. And in this scenario when the topmost layer shifts down again, all the robots gather at one grid vertex (Theorem \ref{thm2}).
\begin{figure}[ht]
    \centering
     \includegraphics[width=0.6\linewidth]{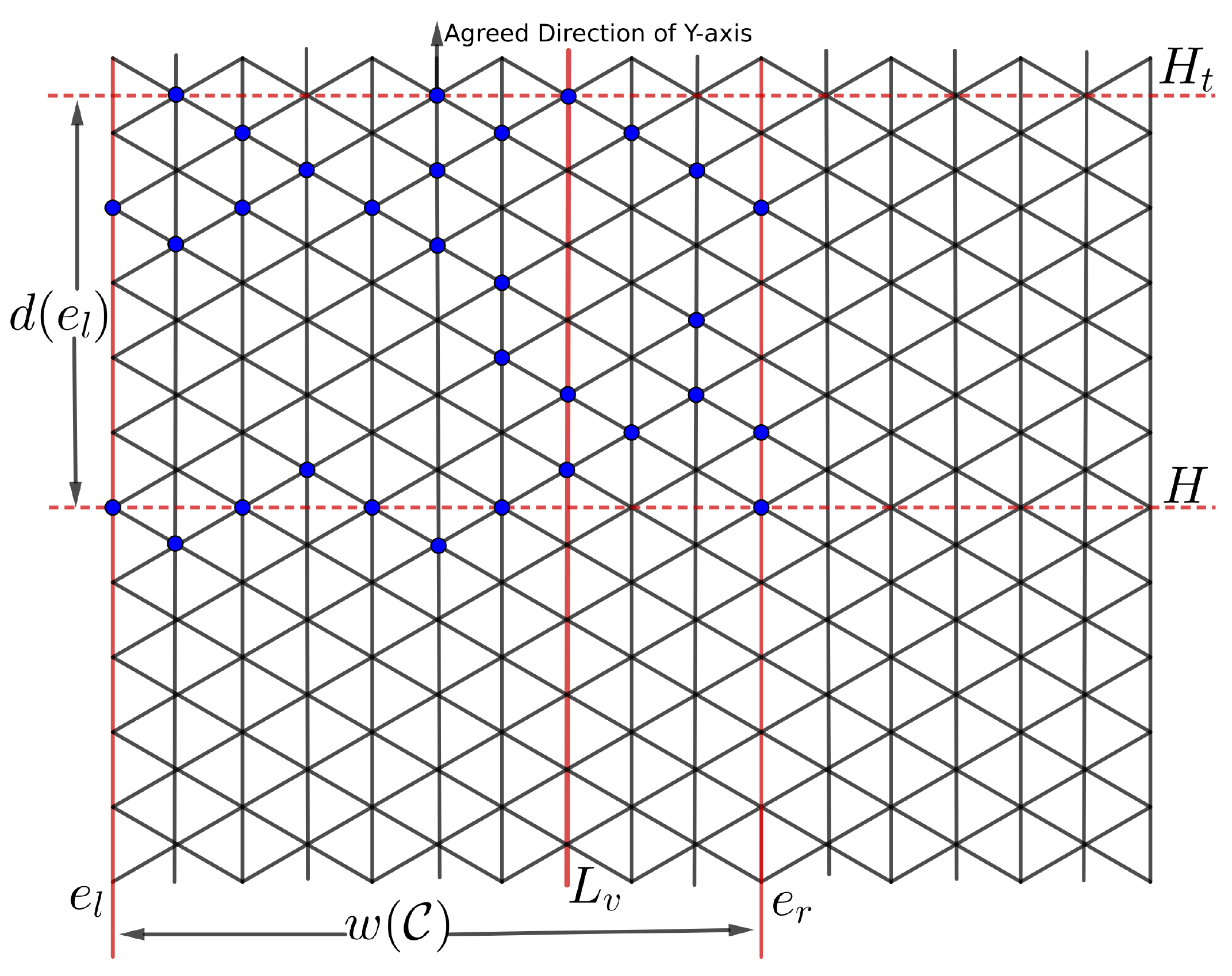}
     \caption{diagram of a configuration $\mathcal{C}$ mentioning layer ($H$), top most layer ($H_t$), vertical line ($L_v$), left edge ($e_l$), right edge ($e_r$), width of the $\mathcal{C}$ ($w(\mathcal{C})$) and depth of the vertical line $e_l$ ($d(e_l)$).}\label{Figdef}
    \end{figure}

\begin{lemma}
\label{flemma1}
If at the start of some round the configuration formed by the robots has connected visibility graph then after execution of Algorithm~\ref{algo1} at the end of that round the visibility graph of the configuration remains connected.
\end{lemma}
\begin{proof}
    Let us consider for some $t >0$, the configuration is denoted as $\mathcal{C}$ at the beginning of round $t$. Then the visibility graph at the beginning of round $t$ is denoted as $G_{\mathcal{C}}$. We will now prove this lemma with the argument that during this specific round $t$, if a robot say $r$, decides to move then no edge in $G_{\mathcal{C}}$ of which $r$ is an end vertex before its move, disappears after the move of $r$. This will imply movement of any robot during the round $t$ does not lose any edge in $G_{\mathcal{C}}$. It is also assumed that for the initial configuration $G_{\mathcal{C}}$ is connected. So, this argument is sufficient to prove that $G_{\mathcal{C}}$ will stay connected throughout the execution of the algorithm~\ref{algo1}. If the vertex on which $r$ is located has more than one robot and at least one robot is not activated during the round $t$, it stays connected to $r$ even if $r$ moves during the round $t$, as even after the move of $r$, the distance from the previous vertex to the vertex $r$ reaches after the move is 1-hop. So without loss of generality let $r$ be singleton on its location. We now have two cases: 

\textbf{Case-I:} Let us consider $r$ is not an \texttt{extreme} robot. And during the round $t$, $r$ decides to move. This implies during the look phase of that round, $r$ has seen at least two robots on both of its $v_2(r)$ positions and no robot with $y-coordinate$ greater than zero. Let $r_1$ be any robot on one of the $v_2(r)$ position and  $r_2$ any robot on another $v_2(r)$ position of $r$. Now there are two sub-cases.

\textbf{Case-I(a):} For the first case, let us consider the case where $r$ does not see any robot on $v_1(r)$ position during the look phase of round $t$. Note that in this case before the move of $r$, $r$ is end vertex of the edges $rr_1$ and $rr_2$. Also note that even if $r_1$ and $r_2$ are \texttt{extreme} and active during the round $t$, they do not move during this round as $r_1$ and $r_2$ see robot $r$ at $v_3(r_1)$ and $v_3(r_2)$ positions respectively and does not see any robot on $v_2(r_1) = v_2(r_2)$ . Also, they do not move if they are not \texttt{extreme} as both of them see $r$ at a position with $y-coordinate$ greater than zero. So in this case $r$ moves to $v_1(r)$ during the move phase of round $t$. Note that $v_1(r)$ is 1-hop away from both the $v_2(r)$ position of $r$ before it moves. So after $r$ reaches $v_1(r)$, $rr_1$ and $rr_2$ are both still edges of $G_{\mathcal{C}}$ this is true for any $r_1$ and $r_2$ on both the $v_2(r)$ positions respectively before $r$ moves.

\textbf{Case-I(b):} For the second case, let us consider there is at least a robot at $v_1(r)$ during the look phase of round $t$. Let $r_3$ be any robot on $v_1(r)$. Then before the move of $r$, it is the end vertex of the edges $rr_1, rr_2$ and $rr_3$. In this case note that if $r_3$ gets activated during round $t$, it can not be \texttt{extreme} as it sees $r$ on its positive $y-$axis. Also, for this reason, $r_3$ does not move during the round $t$. Now if $r_1$ and $r_2$ are not \texttt{extreme} they will not move during round $t$ even if they are activated and $r$ moves to $v_1(r)$ to the location of $r_3$. Now with a similar argument for the above case, we can say $rr_1, rr_2$ and $rr_3$ will still be edges after the move of $r$. So let us consider either $r_1$ or $r_2$ is \texttt{extreme} and activated during the round $t$. Without loss of generality let $r_1$ is \texttt{extreme} and it is activated during round $t$ along with $r$. Now $r_1$ will see  robots either in the positions $v_3(r_1)$ and $v_2(r_1)$ or on the positions $v_3(r_1), v_2(r_1)$ and $v_1(r_1)$. In both of these cases $r_1$ moves to $v_2(r_1)$ to the location of $r_3$ along with $r$. So after the move of $r$, $rr_1,rr_2$ and $rr_3$ are still edges of $G_{\mathcal{C}}$. Hence we can conclude that move of a non \texttt{extreme} robot $r$, does not lose any edge of $G_{\mathcal{C}}$ of which $r$ was an end vertex.

\textbf{Case-II:} Let us consider $r$ is an \texttt{extreme} robot that decides to move during a round $t$. Then There are five possible views of $r$ during round $t$.

 \textbf{Case-II(a):} $r$ only sees robots at $v_2(r)$ during look phase of round $t$. Let $r_1$ be a robot on $v_2(r)$. Note that before $r$ moves, $rr_1$ is an  edge of $G_{\mathcal{C}}$ of which $r$ is an end vertex (for any $r_1$ on $v_2(r)$). Now in this case even if $r_1$ is activated during the round $t$, it either sees only $r$ on $v_3(r_1)$ or sees robots on $v_3(r_1)$ and $v_1(r_1)$ during look phase of the round $t$. For both of the cases, $r_1$ does not move during round $t$. Now $r$ moves to $v_2(r)$ at the location of $r_1$. So it is evident that even after $r$ moves $rr_1$ still is an edge of $G_{\mathcal{C}}$.
 
 \textbf{Case-II(b):} $r$ only sees robots at $v_1(r)$ during the look phase of round $t$. Let $r_1$ be a robot on $v_1(r)$. For any $r_1$ at $v_1(r)$, $rr_1$ is an edge of $G_{\mathcal{C}}$ before $r$ moves. Note that $r_1$ is not an \texttt{extreme} robot and it sees $r$ with $y-coordinate$ greater than zero. So, during round $t$ even if $r_1$ is activated, it  never moves. Now $r$ moves to $v_1(r)$ to the location of $r_1$. So it is obvious that $rr_1$ will still be an edge of $G_{\mathcal{C}}$ even after $r$ moves.
 
 \textbf{Case-II(c):} $r$ only sees robots at $v_3(r)$ and $v_2(r)$ during the look phase of round $t$. Note that before the move of $r$, for any $r_1$ at $v_3(r)$ and $r_2$ at $v_2(r)$, $rr_1$ and $rr_2$ are the edges of $G_{\mathcal{C}}$ of which $r$ is an end vertex. Now if $r_1$ and $r_2$ are not activated at round $t$, then  $r$ moves to $v_2(r)$ at the position of $r_2$. Now since $r_2$ and $r_1$ are only 1-hop distance apart, after the move of $r$, it still will be the end vertices of the edges $rr_1$ and $rr_2$. Note that during the round $t$, even if $r_2$  is activated it never moves as it is not an \texttt{extreme} robot and it sees $r_1$ with $y-coordinate$ greater than zero. So let us now consider the case where both $r$ and $r_1$ are activated during the round $t$. In this case if $r_1$ is not \texttt{extreme} and it moves during round $t$ it moves to the location $v_1(r_1) = v_2(r)$, to the location of $r_2$. Now if $r_1$ is \texttt{extreme} it can see robots only on $v_2(r_1)$ and $v_1(r_1)$. In this case also $r_1$ moves to $v_1(r_1) = v_2(r)$, to the location of $r_2$. $r$ also moves to the location of $r_2$.  So even if $r$ and $r_1$ both moves during the round $t$, after the movement $rr_1$ and $rr_2$ are still edges of $G_{\mathcal{C}}$.
 
 \textbf{Case-II(d):} $r$ only sees robots say at $v_2(r)$ and $v_1(r)$ during the look phase of round $t$. Observe that for any $r_1$ at $v_2(r)$ and for any $r_2$ at $v_1(r)$, $r$ is an end vertex of the edges $rr_1$ and $rr_2$ in $G_{\mathcal{C}}$ before it moves. Note that if $r_1$ and $r_2$ does not move during round $t$, then $r$ moves to $v_2(r)$, at the location of $r_2$. Now since $r_2$ is at 1-hop distance from $r_1$, after the move $r$ still is end vertex of the edges $rr_1$ and $rr_2$ of $G_{\mathcal{C}}$. Observe that even if $r_2$ is activated during the round $t$, it does not move as it sees $r$ with $y-coordinate$ is greater than zero. So let us consider $r$ and $r_1$ both are activated at round $t$ and both decides to move. This implies $r_1$ is \texttt{extreme} and during the look phase of round $t$, $r_1$ either sees robots at the positions $v_3(r_1)$ and $v_2(r_1)$ or sees robots at the locations $v_3(r_1), v_2(r_1)$ and $v_1(r_1)$. For both the views $r_1$ moves to $v_2(r_1) = v_1(r)$ i.e, at the location of $r_2$ . Also $r$ moves to $v_1(r) = v_2(r_1)$ i.e at the location of $r_2$. So even if both $r$ and $r_1$ moves during the round $t$, they both moves to $r_2$ during round $t$. So after their move $rr_1$ and $rr_2$ are still edges of $G_{\mathcal{C}}$.
 
 \textbf{Case-II(e):} $r$ only sees robots say at the positions $v_1(r)$, $v_2(r)$ and $v_3(r)$ during the look phase of round $t$. Note that before it moves, $r$ is end vertices of the edges $rr_1, rr_2$ and $rr_3$ in $G_{\mathcal{C}}$ for any $r_1$ at $v_1(r)$, $r_2$ at $v_2(r)$ and $r_3$ at $v_3(r)$. Now if none of $r_1, r_2$ and $r_3$ are activated  or does not move during the round $t$ , then after $r$ moves to $v_2(r)$ i.e to the location of $r_2$ it is still at most 1-hop distance apart from $r_1, r_2$ and $r_3$. So, $rr_1, rr_2$ and $rr_3$ are still edges in $G_{\mathcal{C}}$. Now observe that even if $r_1$ and $r_2$ are activated they do not move during round $t$ as $r_1$ and $r_2$ both sees $r$ and $r_3$ directly above them (i.e on their respective positive $y-$axis) respectively. So let us consider that only  $r$ and $r_3$ are activated during the round $t$. $r$ moves to $v_2(r)$, at the location of $r_2$. Now if $r_3$ is \texttt{extreme} then it sees $r$ and $r_2$ at the positions $v_2(r_3)$ and $v_1(r_3)$ respectively during the look phase of the round $t$. So $r_3$ moves to $v_1(r_3) = v_2(r)$, i.e the location of $r_2$. Observe that all of $r_1, r_2$ and $r_3$ are still at most 1-hop away from $r$. So even if both $r$ and $r_3$ are \texttt{extreme} and both of them moves during round $t$, $rr_1, rr_2$ and $rr_3$ are still edges of $G_{\mathcal{C}}$. Now let us consider the case where $r_3$ is not \texttt{extreme} but it decides to move during the round $t$. Again both of $r$ and $r_3$ moves to $v_1(r_3) = v_2(r)$ and with the similar argument we can conclude, $rr_1, rr_2$ and $rr_3$ still remains edges of $G_{\mathcal{C}}$ even after both $r$ and $r_3$ moves during round $t$.

 For all the cases and for any robot $r$ that decides to move during a round $t$, we showed that all the edges in $G_{\mathcal{C}}$ of which $r$ is an end vertex before the move does not get disappeared after $r$ moves during the round $t$. Now since the Initial configuration is connected, the graph $G_{\mathcal{C}}$ stays connected in each round. Hence the lemma.\qed 
\end{proof}

\begin{figure}[ht]
    \centering
     \includegraphics[width=0.9\linewidth]{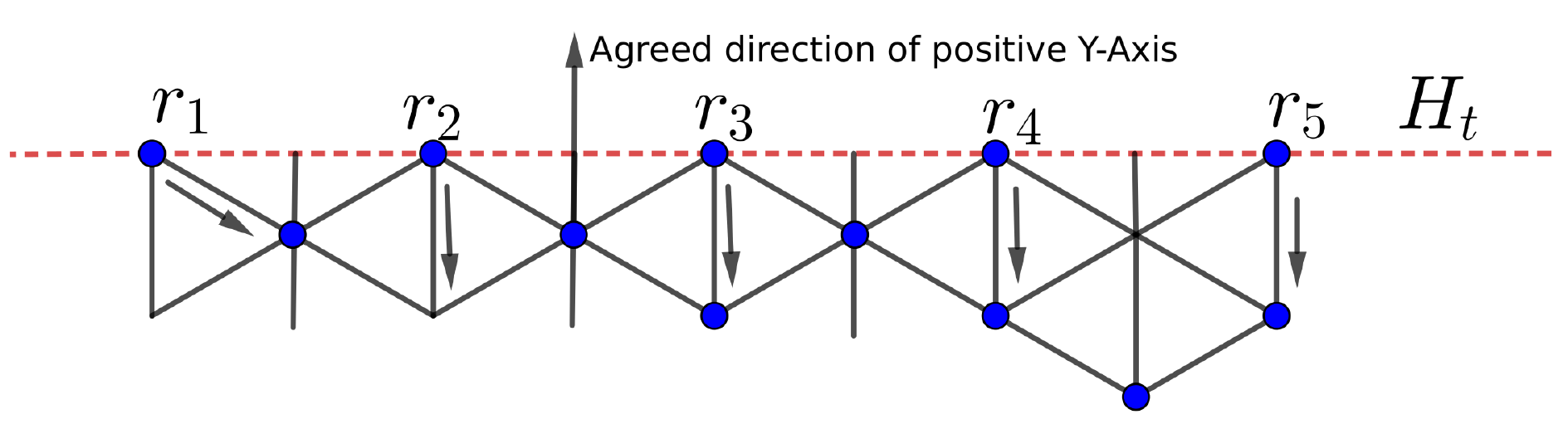}
     \caption{$r_1, r_2, r_3, r_4$ and $r_5$ are robots on $H_t$. Any robot on $H_t$ will always have the view same as one of $r_i$, where $i \in \{1, 2, 3, 4, 5\}$. And for each of these 5 views, a robot always moves to another layer below $H_t$.  }\label{Figtoplayer}
    \end{figure}
    
\begin{lemma}
\label{flemma2}
$H_t$ of the configuration $\mathcal{C}$, always shift down in one epoch until the gathering is complete.
\end{lemma}
\begin{proof}
Let $r$ be a robot on $H_t$. 
If the gathering is not complete then $r$ must see other robots on its adjacent vertices.

Now, there are two cases. 

\textbf{Case-I:} If $r$ is not \texttt{extreme} then upon activation $r$ must see two robots on each of it's $v_2(r)$ position and no robot with $y-coordinate$ greater than zero. So upon activation $r$ moves to $v_1(r)$ which is below $H_t$.

\textbf{Case-II:} If $r$ is \texttt{extreme}, then there are three cases. Firstly  if $r$ sees a robot only on $v_2(r)$ upon activation, then it moves down to $v_2(r)$ which is below $H_t$. Secondly and thirdly, if $r$ sees robot on only $v_1(r)$ or sees robots both on $v_1(r)$ and $v_2(r)$. For both second and third case, $r$ moves to $v_1(r)$ which is also below $H_t$.

Since in one epoch, all robots on $H_t$ must be activated once they must move below $H_t$. Hence, $H_t$ of the configuration $\mathcal{C}$ always shifts down in one epoch. (Fig.\ref{Figtoplayer}).\qed
\end{proof}
\begin{lemma}
\label{flemma3}
Robots  on $e_l$ or $e_r$ which are not \texttt{extreme} do not move.
\end{lemma}
\begin{proof}
Let  $r$ be a robot on $e_l$ or on $e_r$ which is not \texttt{extreme}. Note that $r$ only moves when it sees there is no robot above (i.e no robots with $y-coordinate >0$) and both of its $v_2(r)$ are occupied by some other robots. Now since $r$ is on $e_l$ or on $e_r$, at least one of it's $v_2(r)$ is empty. So $r$ does not move.\qed
\end{proof}

\begin{figure}[ht]
    \centering
     \includegraphics[width=0.6\linewidth]{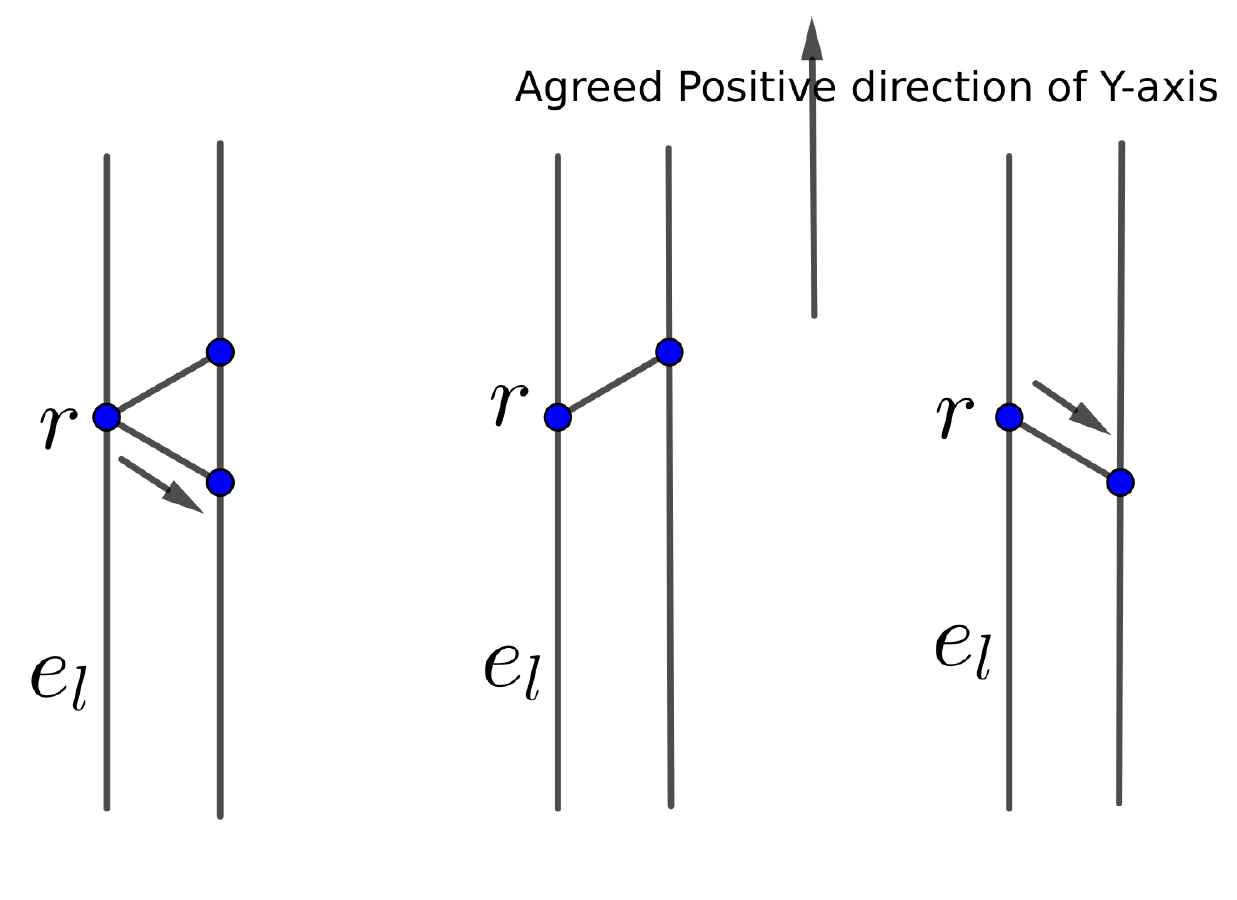}
     \caption{All possible view of the lowest \texttt{extreme} robot $r$ on $e_l$. In each view $r$ never moves directly below to $v_1(r)$.}
     \label{FiglowestExtremeofeL}
    \end{figure}
    
\begin{lemma}
\label{flemma6}
A robot $r$ which is lowest on $e_l$ or $e_r$ never moves down to $v_1(r)$.
\end{lemma}
\begin{proof}
    We will proof this lemma considering $r$ on $e_l$. If $r$ is on $e_r$ the proof will be similar. Let $r$ is the lowest robot on $e_l$. By lemma \ref{flemma3}, if $r$ is not \texttt{extreme} it does not move. Now if $r$ is \texttt{extreme} then no robot will move to $v_1(r)$ from a different vertical line as $v_1(r)$ is empty. So, $r$ never moves to $v_1(r)$ as it can not see any robot on  it's $v_1(r)$ position (Fig \ref{FiglowestExtremeofeL}). So, $r$ will never move to $v_1(r)$.\qed
\end{proof}
\begin{lemma}
\label{flemma4}
Neither $d(e_l)$ nor $d(e_r)$ ever increase as long as the position of the corresponding vertical line is same.
\end{lemma}
\begin{proof}
    We will prove this lemma for $e_l$ only. For $e_r$ the proof will be similar.
Observe that if $d(e_l)$ increase it can not be increased the by the lowest robot (say, $r$) on $e_l$ (by lemma \ref{flemma6}). So the lowest robot can not increase $d(e_l)$. Also, no robot moves above the layer it is on. So, no robot on $e_l$ moves up to increase $d(e_l)$.

Now it might be possible that $d(e_l)$ is increased by a robot that moves below the lowest robot $r$ of $e_l$ from the immediate right vertical line. Note that a robot moves from a vertical line to another vertical line only if it is \texttt{extreme}. Now, an \texttt{extreme} robot, if moves, never go to a position that is not occupied by any other robot before the movement according to algorithm \ref{algo1}(Fig \ref{FigMoveExtreme}). Since, to increase $d(e_l)$, a robot must move to a position that does not contain any other robot before the movement, no robot will come below $r$.

So, $d(e_l)$ never increases. Similarly we can say $d(e_r)$ never increases. Thus the result.\qed
\end{proof}
\begin{lemma}
\label{flemma5}
If $w(\mathcal{C}) > 0$ at a round $t_0$ then there exists a round $t>t_0$ such that $w(\mathcal{C})$ decreases.
\end{lemma}
\begin{proof}
    Let $w(\mathcal{C}) > 0$. Note that no robot from $e_l$ moves left and no robots from $e_r$ moves to its right. So $w(\mathcal{C})$ never increases. Now If possible let $w(\mathcal{C})$ never decreases. This implies the vertical lines $e_l$ and $e_r$ never shifts to right and left respectively.
Now by lemma \ref{flemma4} we can say that $d(e_l)$ and $d(e_r)$ never increases. Also by lemma \ref{flemma2} $H_t$ shifts down always in one epoch. So from these two lemmas we can conclude that there exists $t_1 > t_0$ such that at the round $t_1$ either $d(e_l)$ or, $d(e_r)$ becomes 0. Note that when $d(e_l)$ (or, $d(e_r)$) is 0 then $e_l$ (or, $e_r$) contains exactly one vertex $v$ having robots. Let  $r$ be a robot on that vertex $v$. Note that $r$ is \texttt{extreme} and by Lemma~\ref{flemma1} since $G_{\mathcal{C}}$ is connected $r$ sees  robots only on $v_2(r)$. So, $r$ moves to $v_2(r)$ which is on the next vertical line on its right (or, left). So after a finite epoch either $e_l$ shifts right or $e_r$ shifts left and thus we arrive at a contradiction. Hence the lemma.\qed
\end{proof}
\begin{theorem}\label{thm2}
Algorithm \texttt{1-hop 1-axis Gather} guarantees that there exists a round $t>0$ such that a swarm of $n$  myopic robots on an infinite triangular grid $\mathcal{G}$ with 1-hop visibility and one axis agreement will always gather after completion of round $t$  under semi-synchronous scheduler starting from any initial configuration for which visibility graph $G_{\mathcal{C}}$ is connected.
\end{theorem}
\begin{proof}
    From lemma \ref{flemma5} we can conclude that after a finite number of rounds $w(\mathcal{C})$ becomes zero. Observe that when $w(\mathcal{C)} =0$ all the robots are on a vertical line $e_l$. Also, note that when all the robots are on a single line then $e_l$ is the same as $e_r$. In this situation if $d(e_l) = 0$ that means gathering is complete. So let us assume $d(e_l) > 0$. Note that in this scenario, since no non \texttt{extreme} robot $r$ sees two robots on both of its $v_2(r)$ position and no \texttt{extreme} robot $r'$  sees a robot in a location other than $v_1(r')$, no robot will move to a different vertical line from $e_l =e_r$. Now by lemma\ref{flemma2} and lemma \ref{flemma4} , $H_t$ shifts down until there is only one grid point having robots on $e_l$ (i.e $d(e_l) = 0$).  So we can conclude that there exists a round  $t > 0$ such that gathering is complete after the completion of round $t$.\qed
\end{proof}
\subsection{Complexity Analysis}

First we observe in Theorem~\ref{thm3} that it will take at least $\Omega(n)$ epochs to gather $n$ number of robots. The theorem is stated and proved formally in the following.
\begin{theorem}
\label{thm3}
 Any gathering algorithm on a triangular grid takes $\Omega(n)$ epoch.
\end{theorem} 
\begin{proof}
    Let us consider the configuration in Fig. \ref{Figlowerbound}. Let us define the height of the configuration as the distance between the topmost and lowest layer of the configuration. Note that in Fig.~\ref{Figlowerbound} the height of the configuration is $n$ (i.e., the number of robots). When considering the worst case, a robot can only be activated once in each epoch. Hence, the height of the configuration decreases by at most 2 units in each epoch.  Now the robots will gather when the height of the configuration and width of the configuration both becomes 0. Since in each epoch, height decreases by 2 units, at least $\frac{n}{2}$ rounds will be needed to gather the $n$ robots. Hence the result.\qed
\end{proof}

\begin{figure}[ht]
    \centering
     \includegraphics[height=5cm,width=3cm]{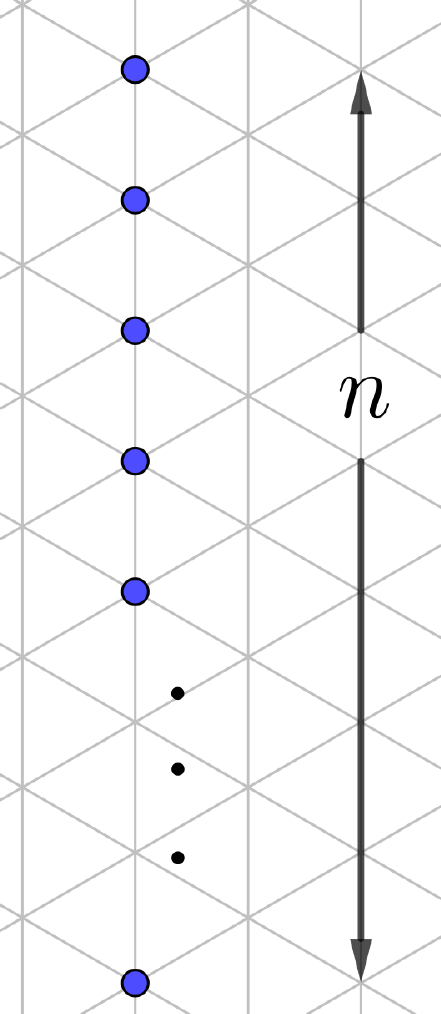}
     \caption{All of the $n$ robots are on a straight line on the triangular grid. Height of the configuration is $n$.}\label{Figlowerbound}
    \end{figure}
    
Now we shall prove that the robots executing our proposed algorithm do not go downwards by much. First, we define the smallest enclosing rectangle for the initial configuration.
\begin{figure}[ht]
    \centering
     \includegraphics[width=0.4\linewidth]{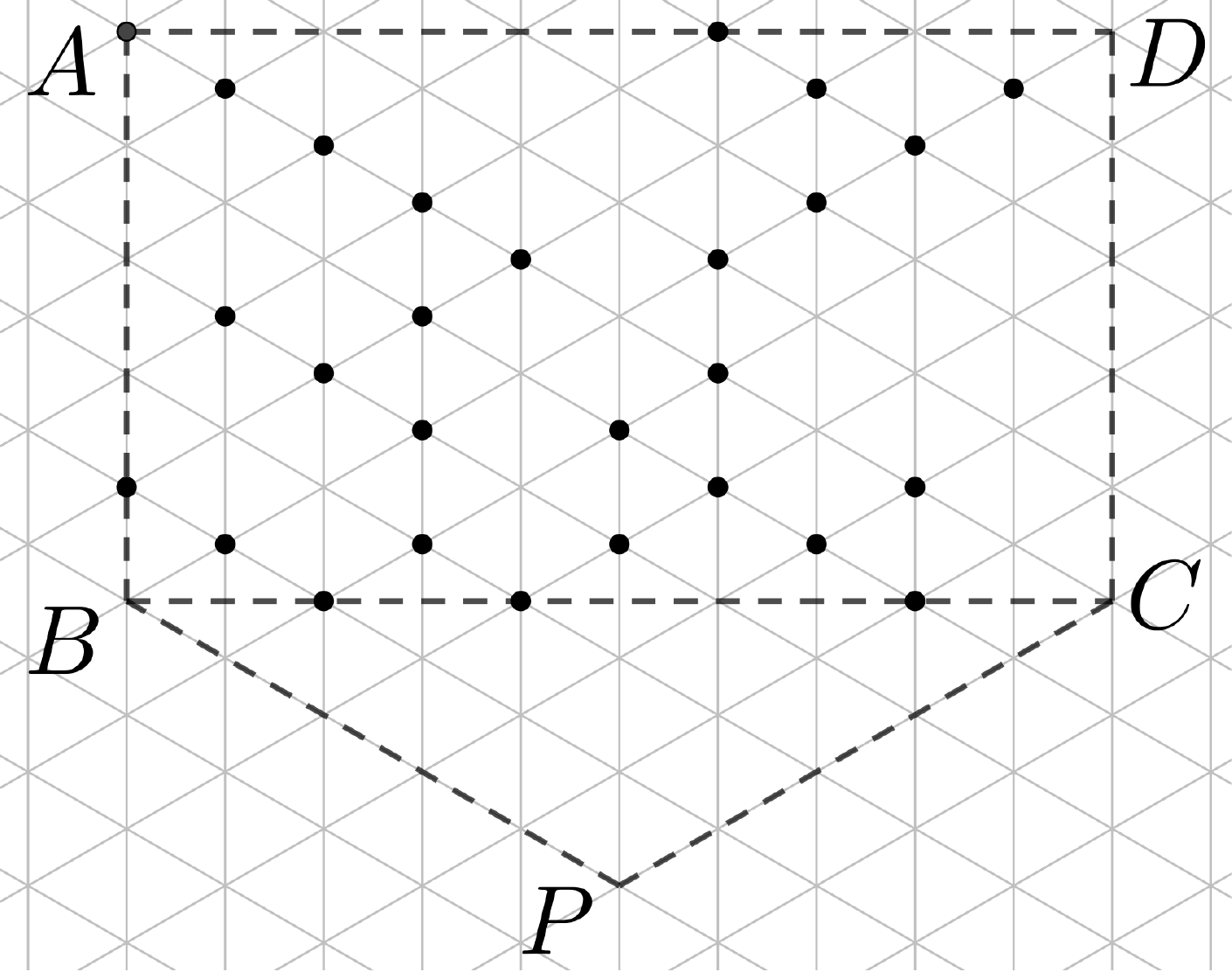}
     \caption{$ABCD$, smallest enclosing rectangle}
     \label{correct1}
    \end{figure}

\begin{definition}[$\mathcal{SER}$]
A rectangle $\mathcal{R}=ABCD$ is said to be the smallest enclosing rectangle ($\mathcal{SER}$) (Figure~\ref{correct1}) of the initial configuration if it is the smallest in dimension satisfying the following:
\begin{enumerate}
    \item All robots in the initial configuration are inside $\mathcal{R}$
    \item All vertices of $ABCD$ are on some grid points
    \item $AB$ and $CD$ side is parallel to the axis agreed by all the robots
    \item $BC$ is the lower side of the rectangle. 
\end{enumerate}
\end{definition}
Next, we define a polygon that shall contain all the robots throughout the algorithm.
\begin{definition}[Bounding Polygon]
Let $\mathcal{R}=ABCD$ be the $\mathcal{SER}$ of the initial configuration. Let $P$ the point below $BC$ line such that $\angle CBP=\angle BCP=\pi/6$. Then the polygon $\mathcal{P}=ABPCDA$ is said to be the Bounding Polygon. 
\end{definition}
  We show that no robot executing Algorithm~\ref{algo1} ever steps out of the bounding polygon (Lemma~\ref{flemma7}). Using Lemma~\ref{flemma7}, Theorem~\ref{thmf} proves that Algorithm~\ref{algo1} terminates within $O(n)$ epochs.
 \begin{lemma}
 \label{flemma7}
 No robot executing the Algorithm~\ref{algo1} ever steps out of the bounding polygon. 
 \end{lemma}
\begin{proof}
    Let $\mathcal{P}=ABPCDA$ be the bounding polygon (Fig. \ref{correct1}). Note that, the point $P$ is on some grid point. Opposite to our claim, let there be some robots that step out of $\mathcal{P}$ and $k^{th}$ round is the earliest round when some robots stepped out $\mathcal{P}$. Let $r$ be such a robot. Note that at the end of $(k-1)^{th}$ round, no robot is outside of $\mathcal{P}$. Firstly since no robot ever moves upward, so no robot can step out of $\mathcal{P}$ through the $AD$ side and goes above to the $AD$ side. Then $r$ can step out of $\mathcal{P}$ through any side of $\mathcal{P}$ but $AD$. Note that, in this case, $r$ must land on that side in some round before. Hence at the start of $k^{th}$ round $r$ must be on that side.  

Let $r$ has stepped out of $\mathcal{P}$ through $AB$ side. If $r$ is on any point but $B$ then to cross the $AB$ line and step out of $\mathcal{P}$ it has to change its vertical line. This is only allowed for an \texttt{extreme} robot according to our algorithm. So if $r$ has to cross the vertical line while stepping out of $\mathcal{P}$ then $r$ must be an \texttt{extreme} robot at $k^{th}$ round. But according to our algorithm, since an \texttt{extreme} robot never occupies an empty grid point, this yields a contradiction. Even if $r$ is at $B$ and steps out of $\mathcal{P}$ without changing its vertical line then it must go down as a non \texttt{extreme} robot. But since the left $v_2(r)$ position of $r$ is empty, so according to our algorithm, $r$ wouldn't go down as a non \texttt{extreme} robot. Hence $r$ can not step out of $\mathcal{P}$ through the $AB$ line. With a similar argument, one can similarly show that $r$ can not step out of $\mathcal{P}$ through the $CD$ line.
 
 Now let the robot $r$ step out of $\mathcal{P}$ through the $BP$ line or $PC$ line. We already showed that $r$ can not be at $B$ or $C$ at the start of $k^{th}$ round. Let $r$ be at some point on $PB$ line or $PC$ line. Now using the same argument as the previous case one can show that $r$ can not be an \texttt{extreme} robot at the start of $k^{th}$ round because $r$ occupies an empty grid point in this round. Now we see $r$ also cannot be a non \texttt{extreme} robot at the start of $k^{th}$ round. Because a non \texttt{extreme} robot only moves when both of its $v_2(r)$ position is nonempty. But this can not be true in this case for robot $r$.
 Hence this shows $r$ cannot step out of $\mathcal{P}$ at all, which contradicts our assumption that some robots have stepped out of $\mathcal{P}$.\qed
\end{proof}
\begin{theorem}
\label{thmf}
The algorithm 1-hop 1-axis Gather takes at most $O(n)$ epochs to gather all the robots.  
\end{theorem}
\begin{proof}
    Let $\mathcal{R}=ABCD$ be the $\mathcal{SER}$ of the initial configuration and $\mathcal{P}=ABPCDA$ be the bounding polygon. Since the total number of robots is $n$, from simple geometry $|BC|\le(n+1)\dfrac{\sqrt{3}}2$ and so $|BC|\le n+1$. Also the maximum distance from $AD$ to $P$ is $\dfrac 5 4(n+1)$. Therefore the maximum distance of $H_t$ of the initial configuration from $P$ is $\dfrac 5 4(n+1)$.

Now from Lemma~\ref{flemma2} we can say that $H_t$ shifts down at least half a unit in one epoch till the gathering is not done. And since the maximum distance of $H_t$ of the initial configuration from $P$ is $\dfrac 5 4(n+1)$, so within $2\times\dfrac 5 4(n+1)=\dfrac 5 2(n+1)$ epochs the gathering must be complete. Hence the result follows.\qed
\end{proof}
\section{Conclusion}
\textit{Gathering} is a classical problem in the field of swarm robotics. The literature on the gathering problem is vast as it can be considered under many different robot models, scheduler models, and environments. Limited vision is very practical when it comes to robot models. To practically implement any algorithm considering a robot swarm having full visibility is impossible. So, we have to transfer the research interest towards providing algorithms that work under limited visibility also. This paper is one achievement towards that goal.

In this paper, we have done a characterization of gathering on an infinite triangular grid by showing that it would not be possible to gather from any initial configuration to a point on the grid if the myopic robots having a vision of 1-hop do not have any axis agreement even under the FSYNC scheduler. Thus, considering one axis agreement we have provided an algorithm that gathers $n$ myopic robots with a vision of 1-hop under the SSYNC scheduler within $O(n)$ epochs. We have also shown that the lower bound of time for gathering $n$ robots on an infinite triangular grid is $\Omega(n)$. So our algorithm is time optimal.

For an immediate course of future research, one can think of solving the gathering problem by considering myopic robots on an infinite triangular grid making the algorithm collision-free (where no collision occurs except at the vertex of gathering) and under an asynchronous scheduler. Another interesting work would be to find out if there is any class of configurations for which gathering on a triangular grid will be solvable even without one axis agreement.
\vspace{2mm}

\textbf{Acknowledgement:} First and second Authors are supported by UGC, the Government of India. The third author is supported by the West Bengal State government Fellowship Scheme.
%
%
%
 \bibliographystyle{splncs04}
 \bibliography{samplepaper}
\end{document}